\documentclass[runningheads]{llncs}
\usepackage[T1]{fontenc}

\usepackage{pifont}
\usepackage[title]{appendix}
\usepackage{tikz}
\usetikzlibrary{decorations.pathreplacing,fit,calc}

\usepackage{todonotes}
\usepackage[noline,boxed]{algorithm2e}
\usepackage[skins]{tcolorbox}
\usepackage{xcolor}
\usepackage{subfig}
\usepackage{enumerate}
\usepackage{tipa}
\usepackage{ctable}
\usepackage[mathscr]{eucal}
\usepackage{threeparttable}
\usepackage{xcolor}
\usepackage{enumitem}
\usepackage{graphicx}
\begin{document}
\title{A Privacy-Preserving Graph Encryption Scheme Based on Oblivious RAM}
%
%
\author{Seyni Kane \inst{2,3}\thanks{The research presented in this paper was conducted while the author was affiliated with IRT-SystemX} \and
Anis Bkakria\inst{1}\orcidID{0000-0002-9758-4617} 
 }
\authorrunning{S. Kane and A. Bkakria}
%
\institute{IRT SystemX, Palaiseau, France\\ \and Applied Crypto Group, Orange Innovation, 14000 Caen, France. \and SAMOVAR, Télécom SudParis, Institut Polytechnique de Paris, France.
}
\maketitle              
\begin{abstract}
Graph encryption schemes play a crucial role in facilitating secure queries on encrypted graphs hosted on untrusted servers. With applications spanning navigation systems, network topology, and social networks, the need to safeguard sensitive data becomes paramount. Existing graph encryption methods, however, exhibit vulnerabilities by inadvertently revealing aspects of the graph structure and query patterns, posing threats to security and privacy. In response, we propose a novel graph encryption scheme designed to mitigate access pattern and query pattern leakage through the integration of oblivious RAM and trusted execution environment techniques, exemplified by a Trusted Execution Environment (TEE). Our solution establishes two key security objectives: (1) ensuring that adversaries, when presented with an encrypted graph, remain oblivious to any information regarding the underlying graph, and (2) achieving query indistinguishability by concealing access patterns. Additionally, we conducted experimentation to evaluate the efficiency of the proposed schemes when dealing with real-world location navigation services.

\end{abstract}

\section{Introduction}
Cloud computing is a paradigm that offers on-demand storage and computing resources to individuals and enterprises. Large data sizes motivate storage outsourcing for reasons of cost, availability, and efficiency. Storing and processing data securely on the cloud can be challenging. To protect outsourced data, Structured Encryption (SE) \cite{chase2010structured} has been proposed. A structured encryption scheme encrypts structured data in such a way that it can be queried through the use of a query-specific token that can only be generated with knowledge of the secret key. In addition, the query process reveals no useful information about either the query or the data. Structured encryption schemes include graph encryption schemes \cite{ghosh2021efficient,meng2015grecs}.

Graph encryption schemes have many applications, such as private navigation systems \cite{wu2016privacy}, online social networks \cite{lai2019graphse2}, modeling highly confidential infrastructure and more. The main challenge is to design graph encryption schemes that are secure, expressive, and efficient. Some cryptographic techniques, such as Fully Homomorphic Encryption (FHE) \cite{gentry2009fully} and secure Multi-Party Computation (MPC), can achieve high security but low efficiency. Other solutions trade off some security for better efficiency by allowing controlled leakage \cite{ghosh2021efficient,meng2015grecs}. One example of leakage is the Access Pattern (AP), which is the set of nodes and edges accessed by a query on the encrypted graph. Another example is the Query Pattern (QP), which is the set of queries issued on the encrypted graph. An adversary who observes the AP and QP can infer information about the graph structure, connectivity, and content, as well as the query frequency, similarity, and user interests.

Ghosh, Kamara and Tamassia (GKT) \cite{ghosh2021efficient} proposed a practical graph encryption scheme that supports shortest path queries with a good trade-off between efficiency and security. They encrypts the graph using a recursive algorithm that partitions the graph into subgraphs and encrypts them separately. It has optimal preprocessing time and space, and query time proportional to those for the unencrypted graph. However, it still leaks AP and QP. Falzon et al. \cite{falzon2022efficient} presented an attack on the GKT scheme that exploits the QP leakage to recover the queries with high probability.

Due to access pattern and query pattern leakage, existing graph encryption schemes are not suitable for highly sensitive graphs. A possible solution is to use Oblivious RAM (ORAM) \cite{goldreich1996software} techniques to hide the AP and QP leakage from an untrusted server, and avoid such as the ones proposed in  \cite{falzon2022efficient}.  ORAM is a technique that allows a client to access data in a data store without revealing which item it is interested in. It does this by accessing multiple items each time and periodically reshuffling some or all of the data in the data store.

However, using ORAM for graph encryption schemes can be inefficient, because it requires the client to interact with the server for each recursive call to fetch the path between two nodes. This results in high communication complexity. A way to overcome this issue is to use Trusted Execution Environments (TEE), which is a technology that allows an application to create a protected area in its address space, called an enclave. The data and instructions inside the enclave are confidential and secure even if the attacker has full control of the operating system. TEE can be used to create more efficient oblivious data structures \cite{fuhry2017hardidx,rachid2020enclave,sasy2017zerotrace,wu2023obi}, by placing a mini-client inside the untrusted server. This reduces the communication complexity between the client and the server.

\subsection{Our Contributions}
In this paper, we introduce an innovative graph encryption scheme designed to facilitate shortest path queries without compromising any information pertaining to the underlying data or the query itself. Our approach builds upon the foundation of the GKT scheme \cite{ghosh2021efficient}, widely recognized as the state-of-the-art solution for graph encryption, and incorporates advanced privacy-preserving techniques such as Oblivious RAM (ORAM) \cite{goldreich1996software} and hardware isolation mechanisms like Trusted Execution Environments (TEE) \cite{costan2016intel}.

By seamlessly integrating ORAM, our scheme effectively eliminates Access Pattern and Query Pattern leakage. Additionally, the incorporation of TEE optimizes communication complexity between the client and the server by establishing a secure enclave within the untrusted server. The key contributions proposed in this work can be resumed as follows:
    \begin{itemize}
        \item We introduce TOGES, the first graph encryption scheme that synergistically employs ORAM and TEE to conceal both AP and QP leakage from an untrusted server. We prove the adaptive semantic security of our scheme, along with the assurance of AP indistinguishability and QP indistinguishability.
        \item We implement our scheme and conduct extensive evaluations on a real location navigation dataset, showcasing its practicality and efficiency.
    \end{itemize}

\subsection{Organization of the paper}
The rest of this paper is organized as follows: Section \ref{sec:related-work} reviews the related work on graph encryption, ORAM, and TEEs. It also discusses recent constructions that combine ORAM and TEE to create more efficient oblivious data structures. Section \ref{sec:preliminaries} introduces the notation, definitions, and background schemes used in this paper, such as GKT, Path ORAM, and TEEs. Section \ref{sec:sys-secu-model} defines the system and security model that we consider for our construction. It also states the assumptions, threat model, and security goals that we aim to achieve. Section 5 describes our construction in two versions: a basic version utilizing ORAM-based graph encryption, and an enhanced version leveraging TEE-ORAM based graph encryption. Section \ref{sec:secu-analysis} analyzes the security of our construction and provides a formal security proof. Section \ref{sec:perfo-analysis} analyzes the performance of our construction. Section \ref{sec:conclusion} concludes the paper and suggests some possible future directions.

\section{Related Work}
\label{sec:related-work}
In this section, we delve into the advancements in the key theoretical foundations underpinning our research. We commence with an exploration of graph encryption, followed by an examination of ORAM. Lastly, we survey pertinent literature that leverages TEE to enhance the efficiency of oblivious data structures.

Several approaches have been proposed for private shortest path computation. Previous studies \cite{mouratidis2013strong,mouratidis2012shortest,xi2014privacy} have delved into PIR-based solutions aimed at concealing the client's location. In the approaches outlined in \cite{mouratidis2012shortest,mouratidis2013strong}, the client initiates a confidential retrieval of specific subregions within the graph that pertain to its query. Following this, the client locally calculates the shortest path within the retrieved subgraph. In the work presented in \cite{xi2014privacy}, the client discreetly solicits columns from the next-hop routing matrix to glean information about the subsequent hop in the shortest path. While these methodologies effectively safeguard the privacy of the client's location, it is essential to note that they do not address the concealment of the server's routing details.

Diverging from prior methodologies, Graph encryption schemes are a type of structured encryption that allows querying encrypted graph. It was introduced by Chase and Kamara \cite{chase2010structured} in their seminal work on structured encryption which is a generalization of searchable encryption to the setting of arbitrarily data structure. They also presented a model and a security definition for graph encryption. Meng et al. \cite{meng2015grecs} proposed three graph encryption schemes that support approximate shortest distance queries on encrypted graphs, using different distance oracles, each with a slightly different leakage profile. Ghosh et al. \cite{ghosh2021efficient} proposed an efficient graph encryption scheme that supports exact shortest path queries on encrypted graphs, using a recursive algorithm based on SP-matrix and non-response revealing Dictionary Encryption Scheme (DES) \cite{cash2014dynamic}. However, their scheme has some leakage that can be exploited by query recovery attacks, as shown by Falzon and Paterson \cite{falzon2022efficient}.

Another family of solutions are those based on Oblivious RAM (ORAM), a technique employed to conceal the access pattern of structured encryption—an area of primary interest. The inception of ORAM can be attributed to Goldreich and Ostrovsky \cite{goldreich1996software}, who pioneered its development in the context of software protection and simulation. Subsequently, numerous practical constructions have emerged, with efforts aimed at enhancing the efficiency of Goldreich and Ostrovsky's initial model, as evidenced by works such as \cite{chan2017oblivious,yeo2021oblivious}. However, a notable drawback of these advancements lies in their substantial client-side storage requirements, which scale proportionally with the size of the underlying data structure.

Addressing the need for more practical and space-efficient ORAMs, Shi et al. introduced recursive ORAMs \cite{shi2011oblivious}. This innovative approach is rooted in tree data structures, where each node functions as a small bucket ORAM. Accessing an element involves traversing a path in the tree, leading to partial reshuffling for enhanced security.

Subsequent contributions, such as those by \cite{moataz2014recursive,stefanov2018path}, further refined and expanded upon the concept of recursive ORAMs. Notably, Path ORAM \cite{stefanov2018path} stands out as one of the most intriguing implementations, currently recognized as among the most efficient ORAM schemes. Path ORAM employs a binary tree structure, accessing data along the path from the root to the leaf. After each access, the data undergoes shuffling and re-encryption to mitigate potential information leakage. The challenge of large client-side storage, exemplified by the position map in Path ORAM, is effectively addressed by incorporating a second ORAM for achieving recursion.

There are recent advancements that integrate Oblivious RAM (ORAM) techniques with Trusted Execution Environments (TEE) to enhance the efficiency of oblivious data structures.Trusted Execution Environments, such as Intel's Software Guard eXtensions (SGX) \cite{anati2013innovative,costan2016intel,mckeen2013innovative}, employ hardware-based approaches to optimize search operations on encrypted data within an untrusted server. Intel's SGX, in particular, has inspired various applications of TEEs.

One application involves constructing secure indexes for encrypted data. For instance, Fuhry et al. introduced HardIDX \cite{fuhry2017hardidx}, an encrypted database index based on B$^+$ trees, implementing search functionality within an SGX enclave. Mishra et al. proposed Oblix \cite{mishra2018oblix}, an oblivious search index that conceals memory accesses and result size while supporting insertions and deletions. Oblix utilizes novel oblivious-access techniques on hardware enclave platforms like Intel SGX.

Another application is private search over encrypted data. Cui et al. developed an SGX-assisted scheme \cite{cui2018preserving} that protects access patterns against side channel attacks without compromising search efficiency. Amjad et al. introduced the first SGX-supported dynamic Searchable Symmetric Encryption (SSE) constructions \cite{amjad2019forward}, ensuring both backward and forward privacy.

Relevant to our work are approaches utilizing oblivious RAM to conceal access patterns, leveraging secure enclaves like Intel SGX for improved efficiency. For instance, Sasy et al. proposed Zerotrace \cite{sasy2017zerotrace}, an enclave-based ORAM scheme offering security against powerful adversaries. Rachid et al. devised efficient techniques for constructing ORAMs using Intel SGX \cite{rachid2020enclave}, implementing and comparing multiple ORAM schemes. Wu et al. introduced OBI \cite{wu2023obi}, a multi-path ORAM demonstrating superiority over traditional single-path ORAM in terms of local stash size and insertion efficiency, while ensuring strong security guarantees.

As previously mentioned, prevalent graph encryption methods, exemplified in \cite{ghosh2021efficient,meng2015grecs}, exhibit vulnerabilities, particularly in Access Pattern (AP) and Query Pattern (QP) leakage, which could be exploited for information extraction. In response, our privacy-preserving graph encryption scheme has been meticulously designed to eliminate both AP and QP leaks. In comparison to existing recursive ORAM techniques \cite{moataz2014recursive,stefanov2018path}, we enhance our ORAM execution's efficiency by introducing a non-interactive approach through server-side TEE integration, thereby fortifying security while achieving superior performance.

\section{Preliminaries}
\label{sec:preliminaries}
In this section, we will introduce some fundamental definitions and cryptographic notions essential for the rest if this paper.

\subsection{Graph}
A graph $G = (V, E)$ consists of a set of vertices $V$ and a set of edges $E$ that connect vertices in pairs. A graph is directed if the edges have a direction from one vertex to another. Two vertices $u, v \in V$ are connected if there is a path from $u$ to $v$ in $G$. We only consider static graphs that support \textbf{Single Pair Shortest Path (SPSP)} queries. A SPSP query $SPath(G, (u, v))$ \cite{ghosh2021efficient}, takes a graph $G = (V, E)$ and two vertices $u, v \in V$ as input and returns a simple path $p_{u, v}$, i.e, a list of nodes $(w_1, ..., w_t)$ such that $(u, w_1), (w_1, w_2), \dots , (w_t, v) \in E$. If there is no path from $u$ to $v$ in $G$, SPath returns $\bot$. 

A \textbf{tree} is a graph that is connected and has no cycle. A rooted tree $T = (V, E, r)$ is a tree with a root vertex $r$. For any rooted tree $T = (V, E, r)$ and vertex $v \in V$, we write $T[v]$ for the subtree of $T$ that has v and all its descendants.

\subsection{SP-matrix \cite{ghosh2021efficient}}
We uses SP-matrix as a data structure to for graphs, in order to be able to perform recursive queries on the graphs.

\begin{definition}[SP-matrix\cite{cormen2001introduction}]\label{def_graphs-iso} 
An SP-matrix structure is a $|V| \times |V|$ matrix $M_G$ that is built from a graph $G = (V, E)$ by running an All-Pairs Shortest Path (APSP) algorithm between all pairs of vertices. We associate the rows and columns of $M_G$ with vertices in $V$ and use $(v_i, v_j)$ to refer to the item at row $v_i$ and column $v_j$. For each pair of vertices $(v_i, v_j)$, we store the first vertex on a shortest path from $v_i$ to $v_j$ at $M_G [v_i, v_j]$, or $\bot$ if there is no path. All the diagonal entries in $M_G$ are set to $\bot$.
\end{definition}

Given $M_G$, we can find the shortest path between two vertices $(v_i, v_j)$ as follows. Look at $(v_i, v_j)$ in $M_G$ to get an item $w$. If $w \neq \bot$, then $w$ is the first vertex on the shortest path from $v_i$ to $v_j$ and we repeat the process with $(w, v_j)$. The recursion ends when we see a $\bot$. The query time is optimal, i.e., it is proportional to the length of the shortest path.

\subsection{GKT}
The GKT scheme \cite{ghosh2021efficient} supports single pair shortest path (SPSP) queries. An SPSP query on a graph $G = (V, E)$ takes as input a pair of vertices $(u, v) \in V \times V$, and outputs a path $p_{u,v} = (u,w_1, \dots , w_t, v)$ such that $(u, w_1), (w_1, w_2), \dots, $ $(w_{t-1}, v) \in E$ and $p_{u,v}$ is of minimal length. SPSP queries may be answered using a number of different data structures.

The GKT scheme makes use of the SP-matrix. For a graph $G = (V, E)$, the SP-matrix $M_G$ is a $|V| \times |V|$ matrix defined as follows. Entry $M_G[i, j]$ stores the second vertex along the shortest path from vertex $v_i$ to $v_j$; if no such path exists, then it stores $\bot$. An SPSP query $(v_i, v_j)$ is answered by computing $M_G[i, j] = v_k$ to obtain the next vertex along the path and then recursing on $(v_k, v_j)$ until $\bot$ is returned. At a high level, the GKT scheme proceeds by computing an SP-matrix for the query graph and then using this matrix to compute a dictionary $SPDX'$. This dictionary is then encrypted using a dictionary encryption scheme (DES) such as \cite{chase2010structured,cash2014dynamic}. To ensure that the GKT scheme is non-interactive, the underlying DES must be response-revealing \cite{cash2014dynamic}. 

\subsubsection{Leakage of the GKT Scheme}
Ghosh et al. \cite{ghosh2021efficient} provide a formal specification of their scheme’s leakage. Informally, the setup leakage of their scheme is the number of vertex pairs in $G$ that are connected by a path, while the query leakage consists of the query pattern (i.e., if and when a query is repeated), the length of the shortest path, and what they refer to as the path intersection pattern (PIP). Given a of sequence of SPSP queries $(q_1, \dots, q_t)$, where $q_i = (u_i, v_i)$, PIP of $q_t$ reveals the intersections of $p_t = SPSP(G, (u_t, v_t))$  with previous shortest paths that have the same destination $v_t$, (see \cite{ghosh2021efficient} for more details).

\subsubsection{Attacks on the GKT Scheme}
Falzon et al. \cite{falzon2022efficient} proposed an efficient query recovery attack against GKT, exploiting query leakage to mount the attack. In their approach, the adversary receives the original graph along with leaked information about specific subsets of queries. They leverage the query leakage within the GKT scheme to execute a Query Recovery (QR) attack against it. This attack, feasible for an honest-but-curious server, necessitates knowledge of the graph $G$. While this might seem like a stringent requirement, it is, in fact, less demanding than the conditions allowed in the security model of \cite{ghosh2021efficient}, where the adversary even has the liberty to choose $G$.

The attack comprises two phases. Initially, there is an offline pre-processing phase conducted on the graph $G$. In this stage, they extract a plaintext description of all its shortest path trees from $G'$. Subsequently, they process these trees and compute candidate queries for each query, utilizing the canonical labels of each tree. A canonical label serves as an encoding of a graph, facilitating the determination of graph isomorphism. The canonical label of a rooted tree can be efficiently computed using the Aho-Hopcroft-Ullman (AHU) algorithm \cite{aho1983data}.

\section{System and Security Models}
\label{sec:sys-secu-model}
In this section, we introduce the system model, system definition,
threats model and security requirements of our construction.

\subsection{System Model}
In the presented scenario, we examine a situation where a client \texttt{C}, constrained by limited resources, seeks to delegate both data storage and computation tasks to an untrusted cloud server \texttt{SP} equipped with a TEE. This two-party system comprises a client, the data owner, and a cloud server, responsible for hosting the data. Crucially, the client places trust in the TEE embedded within the server, treating it as an extension of its own infrastructure.

\subsection{Assumptions and Attacker Model}
As mentioned earlier, our scheme involves two parties: \texttt{C} responsible for encrypting a graph-based dataset and uploading it to the \texttt{CS}. \texttt{C} is the party that submits queries, and assume its trustworthiness.

\texttt{CS} hosts a TEE and stores the encrypted tree ORAM. Additionally, \texttt{CS} loads data into the enclave for query processing. We consider the possibility of an attacker taking control of the operating system on \texttt{CS}. 

TEE provides a secure execution environment by cryptographically safeguarding code and data on an untrusted server. However, it is susceptible to side-channel attacks. These attacks enable adversaries to extract sensitive information by observing the effects of processing without direct access to the information source. Although the operating system (OS) is untrusted, it still manages the enclave's resources, allowing it to monitor the enclave's behavior. Specifically, the OS can generate a precise trace of the enclave's code and data accesses at the page granularity. This trace may be exploited later to deduce information about the outsourced data and/or executed SPSP queries. It's important to note that we do not address these side-channel attacks in our work. Thus we assume an adversary who cannot extract information from the TEE.

Conversely, the adversary has visibility into all communications between the TEE and the external entities, including the untrusted domain and \texttt{C}.

Particularly, the TEE is supposed to manage the position map and handles queries from \texttt{C} without disclosing sensitive information to \textbf{CS}. It also establishes a secure channel with \texttt{C} to protect their communication. We presume the TEE's trustworthiness and consider the data and data access within the TEE as secure. We Assume \textbf{CS} to be an  honest but curious (semi-honest) entity, meaning it follows the protocol while attempting to acquire information about the outsourced data.

\subsection{Security Goals}
We are interested in building a privacy preserving graph encryption scheme that meets the following security requirements:
\begin{enumerate}
        \item Given an encrypted graph, an adversary cannot learn any information about the underlying graph.
        \item Given the view of a polynomial number of query executions for an adaptively generated sequence of queries $q = (q_1, \dots, q_n)$, an adversary cannot learn any information about $q$ and its access pattern, except the size of the respective path.
\end{enumerate}

\subsection{Syntax}
Our construction is defined by three algorithm \texttt{Setup}, \texttt{Query}, and \texttt{Reveal}, and is described in three versions (a trivial, enhanced, and more enhanced version). In all version of \texttt{OBGE} the \texttt{Setup} and the \texttt{Reveal} algorithms are executed by \texttt{C}. And the \texttt{Query} algorithm is executed by \texttt{C} in collaboration with \textbf{CS}.
\begin{itemize}
    \item $\texttt{Setup} (G = (V, E), \lambda, P, \texttt{SKE})$: is PPT algorithm that takes as input a gaph $G$, a security parameter $\lambda$, a PRF $P$, and symmetric-key encryption \texttt{SKE}. It outputs an ecrypted graph in a tree ORAM $T$ data structure, a position map $PM$, and keys $key := (K_1, K_2, K')$. The keys are generated during the key generation sub-protocol of the setup algorithm. $K_1$ and $K_2$ are for \texttt{SKE}, and $K'$ is for $P$.
    \item $\texttt{Query} (q = (u, v), T, PM, K_2, K')$: is PPT algorithm that takes as input shortest path query $q = (u, v)$, an encrypted tree ORAM $T$, a position $PM$, and keys $(K_2, K')$. It outputs an encrypted path $resp$.
    \item $\texttt{Reveal}(resp, K_1)$: is PPT algorithm that takes as input an encrypted shortest path $resp$, and the key $K_1$. It outputs the decrypted shortest path $p_{u, v}$. 
\end{itemize}

\subsection{Security Definition}
We adopt the standard security definition for graph encryption schemes \cite{chase2010structured,meng2015grecs,ghosh2021efficient}, that follows the real/ideal simulation paradigm and is parameterized by a leakage function $\mathcal{L}$ that formalizes the leakage of the scheme.\\

\begin{definition}
    Let $\Pi = (Setup, Query, Reveal)$ denote a graph encryption scheme with respect to the above syntax, and $\lambda$ a security parameter. Let  $\mathcal{A}$ a PPT stateful adversary, $\mathcal{CH}$ his challenger, and $\mathcal{S}$ a stateful simulator that gets the leakage functions $\mathcal{L}$. We consider the probabilistic experiments \textbf{Real}$_{\Pi, \mathcal{A}}(\lambda)$ and \textbf{Ideal}$_{\Pi, \mathcal{A},\mathcal{S}}(\lambda)$ described as follows:\\

    \paragraph{\textbf{Real}$_{\Pi, \mathcal{A}}(\lambda)$ :}
        
        \begin{itemize}
            \item The challenger $\mathcal{CH}$ runs the key generation procedure from the \texttt{Setup} algorithm, and generate $key := (K_1, K_2, K')$ for symmetric-key encryption and the pseudo-ramdom function.
            \item $\mathcal{A}$ chooses a graph $G = (V, E)$ and sends it to $\mathcal{CH}$ who encrypt it using $K_1$ and output an encrypted tree ORAM $T$.
            \item $\mathcal{A}$ makes a polynomial number of adaptive queries $q = (q_1, \dots, q_n)$, $n \in \mathbf{N}, n \leq Poly(\lambda)$. For each query $q_i$, $\mathcal{A}$ receives $p_i$ and a token ($tk_i \gets P_K'(q_i)$) from the challenger $\mathcal{CH}$ as the transcript of the execution of the \texttt{Query} and \texttt{Reveal} algorithm. Finally, $\mathcal{A}$ returns a bit $b$ as the output by the experiment.
        \end{itemize}
    
    \paragraph{\textbf{Ideal}$_{\Pi, \mathcal{A},\mathbf{S}}(\lambda)$ :}
        \begin{itemize}
            \item $\mathcal{A}$ chose graph $G = (V, E)$.
            \item  Given $\mathcal{L}(G)$, $\mathcal{S}$ output an encrypted tree ORAM $T$, and send it to $\mathcal{A}$.
            \item  $\mathcal{A}$ makes a polynomial number of adaptive queries $q = (q_1, \dots, q_n), n \in \mathbf{N}, n \leq Poly(\lambda)$. For each query $q_i$, $\mathcal{S}$ given $\mathcal{L}(G, q_i)$ returns token $tk_i$ and $p_i$ to $\mathcal{A}$. By simulating the execution of the \texttt{Query} and \texttt{Reveal} algorithms with $\mathcal{S}$ playing the role of \texttt{C}, and  $\mathcal{A}$ playing the role of the server. Finally, $\mathcal{A}$ returns a bit $b$ as the output by the experiment.
        \end{itemize}
    where $b = 1$ indicates that $\mathcal{A}$ believes it is interacting with the real experiment, and $b = 0$ indicates otherwise.
    We say that $\Pi$ is $\mathcal{L}$-secure against adaptive chosen-query attacks if for all PPT adversaries $\mathcal{A}$, there exists a PPT simulator $\mathcal{S}$ such that 
        \begin{equation}
            |PR [\textbf{Real}_{\Pi, \mathcal{A}}(\lambda)] - PR[\textbf{Ideal}_{\Pi, \mathcal{A},\mathcal{S}}(\lambda)]| \leq negl(\lambda)
        \end{equation}    
\end{definition}

\section{ORAM Based Graph Encryption (\texttt{OBGE})}\label{sec:PPGE}
In this section, we will present our construction: a trivial version without recursion, which serves as a model for the scheme, an enhanced version that uses TEEs, and the enhanced version which is a recursive version of the enhanced version.

\subsection{A First Construction}\label{subsec2}
Here we give the description of our trivial construction followed by the algorithms with their detailed explanation. Our trivial construction uses a binary tree storage on the server as in Path ORAM, to eliminate the access and query pattern in the GKT schemes.

\subsubsection{Description of the Protocol}

We propose a novel protocol \texttt{OBGE = (Setup, Query, Reveal)} for privacy-preserving graph encryption based on GKT \cite{ghosh2021efficient}, and Path ORAM \cite{stefanov2018path}. Our protocol allows a client to outsource a graph to a server and query it efficiently without revealing any information about the graph structure or the query results. We use a symmetric-key encryption scheme \texttt{SKE = (Gen, Encrypt, Decrypt)} similar to GKT  and a pseudorandom function $P$ in our protocol. 

The setup phase includes The key generation procedure, that produces $K_1$ and $K_2$ for \texttt{SKE}, and $K'$ for the PRF $P$. In the setup phase, the client encrypts the graph and sends it to the server. 
In the query phase, we use the \texttt{Access} protocol in Path ORAM \cite{stefanov2018path} as sub-procedure in order to be able to perform the oblivious access to the ORAM structure. To do so the client sends a query token the ORAM controller which send back the corresponding encrypted path. The reveal phase is just a decryption of the encrypted path returned by the query phase.

\paragraph*{Setup} The setup process is depicted in Algorithm \ref{alg:1}. Suppose we have a graph $G = (V, E)$, where $V$ is the set of vertices and $E$ is the set of edges. \texttt{C} creates an SP-matrix $M_G$, which stores the shortest paths between any pair of vertices in $G$. Then, the client performs the same steps as in GKT to obtain a dictionary $SPDX$, which maps each pair of vertices $(u, v)$ to the next vertex $w$ on the shortest path linking $u$ to $v$. The client then transforms $SPDX$ into a new dictionary $SPDX'$ as follows. For each entry $(u, v) \mapsto w$ in $SPDX$, the client generates two tokens: one for the key, $tk \gets P_{K'}(u, v)$, and one for the value, $tk' \gets P_{K'}(w, v)$, using a pseudorandom function $P$. The client also encrypts the value using a symmetric key encryption scheme, $ct' \gets \texttt{SKE.Encrypt}(K_1, (w, v))$, where $K_1$ is a secret key. The client then sets $SPDX'[tk] := (tk', ct')$ to construct the new dictionary.

The client creates a binary tree ORAM $T$ of size $N$, where $ZN > |V|^2$, and each node acts as a bucket capable of holding up to $Z$ blocks. A block is represented by the tuple $(tk, (tk', ct'), x)$, where $tk$ is the token identifying the block (linked to the query $(u,v)$), $tk'$ is the token identifying the next hop on the shortest path between $u$ and $v$, $ct'$ is the encryption of $(w, v)$ by SKE, and $x$ denotes the leaf node identifier in the ORAM tree $T$. The client initializes an empty stash $ST$, an array capable of holding blocks, and a position map $PM$, consisting of two columns. The first column stores tokens identifying blocks, and the second column stores the leaf identifier currently associated with each block. For each token $tk \in SPDX'$, the client randomly assigns a leaf identifier $x \gets_R {0, 1, \dots 2^L-1}$ and records it in the position map $PM$.

\RestyleAlgo{boxruled}
\begin{algorithm}[H]\label{alg:1}
        \SetKwComment{Comment}{$\triangleright$\ }{}
	\SetKwInOut{Input}{input}
	\SetKwInOut{Output}{output}
	\SetKw{KwB}{break}
	\SetKw{KwH}{halt}
	\DontPrintSemicolon
	\caption{Setup}
	\Input{Graph $G = (V, E)$, security parameter $\lambda$, pseudo random function $P$, symmetric-key encryption \texttt{SKE}}
	\Output{Encrypted tree ORAM $T$, Position map $PM$, keys $key := (K_1, K_2,  K')$}
	\BlankLine
	Initialize empty dictionary $SPDX, SPDX'$;\;
        Initialise empty position map $PM$, and empty stash $ST$;\;
        $sp \gets \texttt{SKE.Gen}(\lambda)$;\Comment*[r]{Generate keys}
       $K_1 \gets \texttt{SKE.Gen}(sp); K_2 \gets SKE.Gen(sp); K' \gets_{R} \{0, 1\}^{\lambda}$;\;
       $key := (K_1, K_2, K')$;\;
        $SPDX := \texttt{ComputeSPDX}(G)$;\;
        \For{$(lab, val) \in SPDX$}
            {$tk_{lab} \gets P_{K'}(lab)$; $tk_{val} \gets P_{K'}(val)$; \Comment*[r]{Encrypt labels and values}
            $SPDX'[tk_{lab}]:= (tk_{val}, ct)$, where  $ct = \texttt{SKE.Encrypt}_{K_1}(val)$;\;
            }

        \For{$tk \in SPDX'$}
            {$PM[tk]:= x$, where $x \gets_R \{0, 1, \dots, 2^L -1\}$; \Comment*[r]{Assign random positions} } 
        
        \For{$tk \in PM$}
            {$(tk', ct') \gets SPDX'[tk]$; \Comment*[r]{Retrieve encrypted values}
            $c \gets SKE.Encrypt_{K_2}(block)$, where $block = (tk, (tk', ct'), x)$;\;
            Upload $c$ on $\mathcal{P}(x)$; \Comment*[r]{Encrypt and upload encrypted block} }
	\Return $T, PM, ST, key := (K_1, K_2, K')$;
\end{algorithm}

\texttt{C} begins by initializing the tree $T$. For each $block = (tk, (tk', ct'), x)$, employ \texttt{C} to encrypt the block using $K_2$: $c \gets \texttt{SKE.Encrypt}(K_2, block)$. Next, retrieve the corresponding leaf identifier $x := PM[tk]$ from the position map. Subsequently, upload the encrypted block to the appropriate node along the path from the root to the leaf $x$ in $T$, denoted as $\mathcal{P}(x)$. This process follows the same mechanism as described in PathORAM \cite{stefanov2018path}. Consequently, by the end of the setup phase, buckets containing fewer than $Z$ data blocks are populated with dummy blocks.

The client finishes the setup phase by uploading the ORAM tree $T$ to the server. The client stores the position map $PM$, the stash $ST$, and the keys. The \texttt{Setup} protocol is described in detail in Algorithm \ref{alg:1}.

\paragraph*{Query}
To query the shortest path between nodes $u$ and $v$, the client computes the token $tk$ of the query using $P_{K'}(u, v)$.
And execute \texttt{Access}$(tk)$ with $tk$, which return the decrypted $block = (tk, (tk', ct')) \gets \texttt{SKE.Decrypt}_{K_2}(c)$.
If  $block \neq \bot$ \texttt{C} parse $block$ store the ciphertext $ct'$ in a variable $resp$, and recall \texttt{Access}$(tk')$, this time with token of the next hope $tk'$. Repeats this process from step $4$ until the \texttt{Access} procedure return $\bot$. The \texttt{Query} protocol is described in detail in Algorithm \ref{alg:2}.\\

\begin{algorithm}[H]\label{alg:2}
        \SetKwComment{Comment}{$\triangleright$\ }{}
	\SetKwInOut{Input}{input}
	\SetKwInOut{Output}{output}
	\SetKw{KwB}{break}
	\SetKw{KwH}{halt}
	\DontPrintSemicolon
	\caption{Query}
	\Input{Query $q = (u, v)$, Encrypted tree ORAM $T$, $K_2, K'$}
        \Output{Encrypted path $resp$}
	\BlankLine
        
        \texttt{C} computes $tk \gets P_{K'}(q)$;  \Comment*[r]{Extract the corresponding leaf identifier} 
        $x := PM[tk]$, $status := \epsilon$, $resp := \epsilon$; \Comment*[r]{Initialization} 
        Set variable $curr := tk$; \Comment*[r]{Set current node to the leaf} 
        \While{$status \neq "SearchEnd"$ }
            {
                 \texttt{C} executes $block \gets \texttt{Access}(curr)$; \Comment*[r]{Access current block}
                 \eIf{$block = \bot$}
                    {set $status := "SearchEnd"$;}
                    {Parse $block$ as $(tk, (tk', ct'), x)$;\;
                    set $\quad resp := resp \cup ct'$ and $curr := tk'$;}
            }
        \Return $resp$; \Comment*[r]{Return encrypted path}
\end{algorithm}

\paragraph*{Reveal}
Once the client receive the cipher-text he decrypts the ciphertext with his key $K_1$ and obtains the plaintext, which is the requested path. The \texttt{Reveal} protocol is described in detail in Algorithm \ref{alg:3}.\\

\begin{algorithm}[H]\label{alg:3}
        \SetKwComment{Comment}{$\triangleright$\ }{}
	\SetKwInOut{Input}{input}
	\SetKwInOut{Output}{output}
	\SetKw{KwB}{break}
	\SetKw{KwH}{halt}
	\caption{Reveal}
	\Input{$resp$, $key:= K_1$}
        \Output{Decrypted path $p_{u, v}$}
	\BlankLine
        \texttt{C} Set variable $p_{u,v} := \epsilon$ \Comment*[r]{Initialization}
        Parse $resp$ as $ct_1 \cup ct_2 \cup \dots \cup ct_k$ \Comment*[r]{Parse encrypted path}
        \For{$i = 1, \dots k$ }
            {
                $m_i \gets \texttt{SKE.Decrypt}(K_1, ct_i)$ \Comment*[r]{Decrypt each node}
                Set $m := m \cup m_i$ \Comment*[r]{Aggregate decrypted nodes}
            }
	\Return $M$;
\end{algorithm}

\begin{theorem}[Correctness of \texttt{OBGE}]\label{thm1}
If $P$ is a secure \texttt{PRF}, \texttt{SKE} is correct and the \texttt{ORAM} is correct, then \texttt{OBGE} is correct.
\end{theorem}
\begin{proof}
The correctness follows from the fact that the tokens generated by $P$ have a negligible probability of colliding in a random function, and consequently, this negligible probability extends to the PRF $P$ as well.
\end{proof}

\subsubsection{Complexity Analysis}

\paragraph*{Space Complexity}
The space complexity analysis encompasses both server and client storage requirements.
\begin{itemize}
    \item Server Storage. The server's storage entails maintaining a binary tree structure with a depth of $L = \lceil \log_2 N \rceil$ and $2^L$ leaves, where $N$ denotes the number of nodes in the tree. Each node, known as a bucket, can accommodate up to $Z$ real blocks. Consequently, the storage demand on the server side is $\mathcal{O}(ZN)$, which simplifies to $\mathcal{O}(N)$ due to the constant nature of $Z$.
    \item Client Storage. On the client side, storage involves retaining secret keys $K_1$, $K_2$, and $K'$ for symmetric-key encryption (\texttt{SKE}) and the Pseudo-Random Function (PRF) $P$. Additionally, the client manages the stash $ST$ and the position map $PM$. As per \cite{stefanov2018path}, the stash typically occupies $\mathcal{O}(\log N) \cdot w(1)$ blocks with high probability. The position map, comprising $NL = N \log N$ bits, translates to $\mathcal{O}(N)$ blocks.
\end{itemize}

\paragraph*{Communication complexity}
The client communicates with the server for each query $q = (u, v)$ by reading and writing a path of $Z \log N$ blocks for each level of recursion. The number  of recursion is equal to the length of the path $p_{u,v}$. Therefore, the total bandwidth used per query is $2|p_{u,v}|Z \log N$ blocks. Since $Z$ is a constant, and $|p_{u,v}|$ is at most $N$, the bandwidth usage is $\mathcal{O}(N \log N)$ blocks.

\paragraph*{Computation}
The server acts as a storage device, so it only retrieves and stores $\mathcal{O}(\log N)$ blocks per level of recursion. The computation is done by the client. The client's computation is $\mathcal{O}(N \log N) \cdot w(1)$ per query. In practice, most of this time is spent on encrypting and decrypting $\mathcal{O}(\log N)$ blocks per level of recursion. \\

Our basic construction is a straightforward adaptation of Path ORAM \cite{stefanov2018path} with the GKT graph encryption scheme \cite{ghosh2021efficient}. The client stores the stash $ST$ and the position map $PM$. However, this can consume a lot of storage space on the client side, especially for large graphs. This goes against our goal, which was to reduce the client's storage and processing burden by delegating it to the cloud. Moreover, the protocol is interactive and requires communication between the client and the server for every level of recursion, which increases the communication overhead as we can see above. 

\subsection{Enhanced Construction}\label{subsubsec2}

Recognizing the inadequacies of the basic construction when dealing with large graphs (e.g., location navigation graphs), we introduce a substantial enhancement aimed at employing a dual-pronged strategy to effectively tackle these challenges.

Initially, we implement a Trusted Execution Environment (TEE) on the server side, establishing a secure enclave that furnishes a confidential and tamper-resistant execution environment. Within this TEE enclave, a compact client module is deployed to oversee the management of the client's stash ($ST$) and position map ($PM$). This enclave functions as the Path ORAM controller, seamlessly handling client queries and furnishing corresponding results devoid of sensitive data exposure. By adopting this framework, we obviate the necessity for client-side storage and streamline the client-server interaction, thus mitigating security vulnerabilities associated with client-side operations.

Despite the discernible enhancements, the storage capacity of the TEE enclave remains a pivotal concern. Traditional TEE implementations, exemplified by Intel SGX, offer a finite internal storage capacity, typically set at 128 MB \cite{will2023intel}. Although the position map and stash are usually smaller than the graph, they may still surpass the storage threshold for larger graphs. Notably, the combined sizes of the position map and stash may exceed the available storage capacity, presenting scalability hurdles.

To fortify our construction and alleviate the inherent storage limitations of TEEs, we propose a secondary approach. Leveraging a recursive ORAM data structure \cite{cryptoeprint:2017/964}, meticulously engineered to curtail the volume of data stored within the TEE's internal memory, we aim to optimize TEE storage utilization while upholding the security assurances of the ORAM protocol. This stratagem not only enhances the scalability of our construction but also ensures the judicious employment of TEE resources, rendering it conducive for applications entailing expansive graphs and datasets.

\section{Security Analysis}
\label{sec:secu-analysis}
In the following, we examine the security of \texttt{OBGE}. We demonstrate that our construction is adaptively-secure
with respect to a well-defined leakage profile. We define our leakage profile below.

\paragraph{Setup Leakage} Our scheme has no setup leakage. Indeed, in the \texttt{Setup} phase of our scheme, we set the size of the ORAM tree $ZN > |v|^2$, so we do not reveal the size of the graph.

\paragraph{Query Leakage} In the query phase, the server only learns the leaf identifiers associated with blocks containing the path of the current query. And we replace these leaf identifiers with random ones on the tree $T$ as we access the blocks. Moreover, we re-encrypt the accessed blocks using a semantically secure SKE and relocate them either in the stash or in a new branch on the tree associated with the new leaf identifier. The server also does not learn the access pattern, which is guaranteed by the underlying Path ORAM protocol. So the only thing the server learns during a query is the length of the requested path $|p_{u, v}|$, by counting the number of recursions.\\

We can then formalize the leakage function of our scheme as $\mathcal{L} = \mathcal{L}_Q (q = (u, v)) = |p_{u, v}|$. We summarize the security of our scheme in the next theorem.

\begin{theorem}[Security of \textit{OBGE}]\label{thm1}
If $P$ is a secure PRF, \texttt{SKE} is INDCPA-secure, ORAM is secure according to definition in \cite{stefanov2018path} then \texttt{OBGE}, as described above, is adaptively $\mathcal{L}$-semantically secure, where $\mathcal{L}$ is the leakage function. 
\end{theorem}

\begin{proof}
We define a simulator that works as follows. Given $\mathcal{L}$, the simulator generates keys $key := (K_1, K_2, K_3)$, constructs an ORAM tree $T$ that holds $N$ random blocks and a position map $PM$ as specified in the real \texttt{Setup}. For each query $q = (u, v)$, the simulator receives $\mathcal{L}_Q(q)$. It checks the query pattern to see if the query is new. If not, it returns the token it generated before. Otherwise, it randomly picks a token $tk$ that is not previously generated and a random path $p$ and outputs $(tk, p)$. 

To show that this simulator satisfies our security definition, we first argue that the random choices of $tk$ and $p$ will be indistinguishable from the outputs of $P$ and $Reveal$. Then we argue that since all entries of $T$ and $PM$ are random, $(tk, p)$ will be randomly distributed and thus indistinguishable from the output of the PRF $P$ and $Reveal$. The result follows directly.
\end{proof}

\section{Evaluation}
\label{sec:perfo-analysis}
\subsection{Dataset}
For our evaluation, we utilized a real world graph dataset constructed from the OpenStreetMap data for navigating Paris \cite{paris-geo}. This dataset represents the city's infrastructure as a network of interconnected nodes and edges. The nodes in our graph dataset consist of road intersections, providing a detailed representation of the city's street layout. Additionally, points of interest, which encompass various locations searchable within Paris, serve as supplementary nodes. The total number of nodes in the constructed graph is $102,037$.
\subsection{Experimental Setup}
We developed SGX-based implementations of the proposed path ORAM and recursive path ORAM, leveraging Intel's SGX SDK version 2.2 and Intel's IPP cryptographic library. Our experiments were conducted on a machine equipped with 32 GB of RAM and an Intel(R) Core(TM) i7-6770HQ 2.60GHz CPU, with a maximum enclave size of 128 MB. Throughout our tests, we maintained Z=5, representing the number of blocks within each node of the recursive path ORAM.
\begin{figure}[h]
\centering
\includegraphics[scale=0.4]{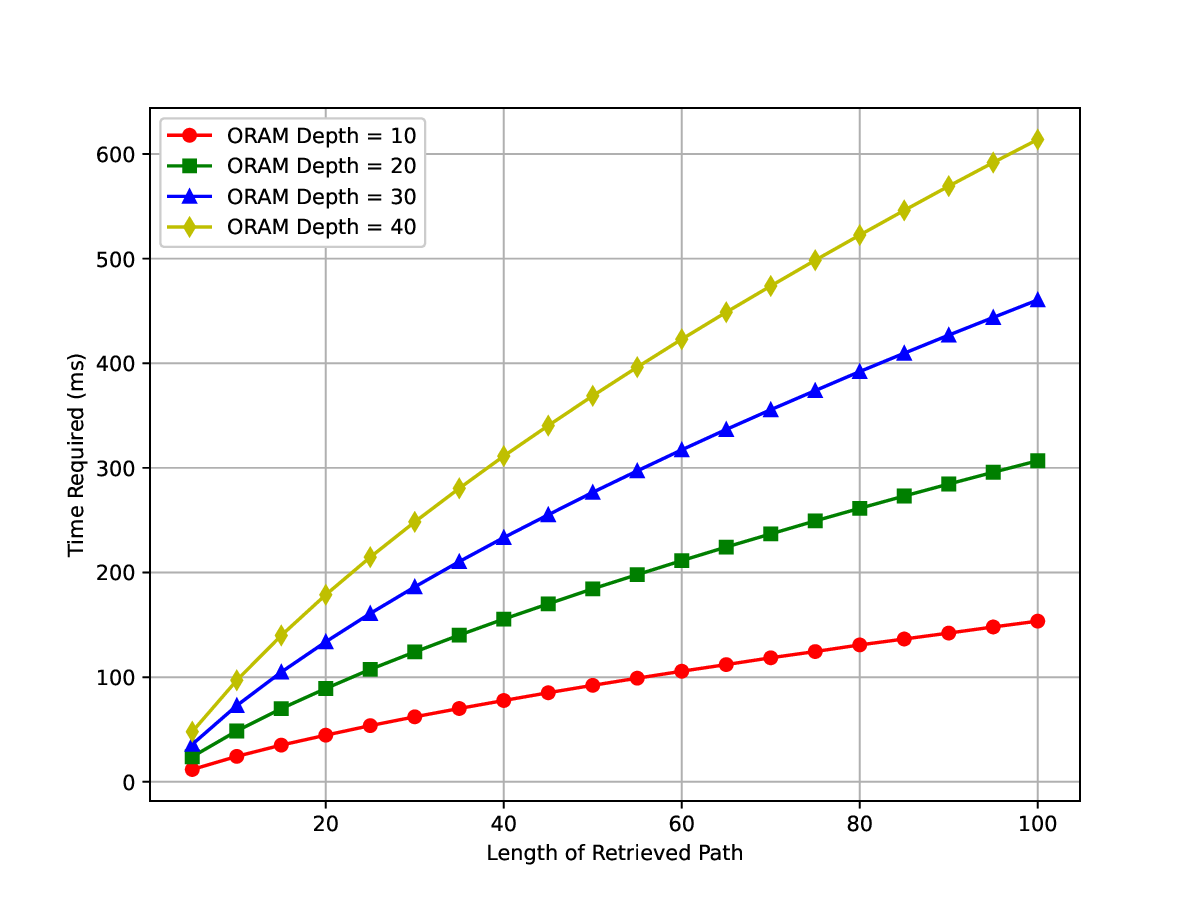}
\caption{Time Required for Path Query Response in Varying recursive ORAM Depths}
\label{fig:time}
\end{figure}

Figure \ref{fig:time} reports the correlation between the length of the retrieved path and the time required for query response but also underscore the practicality of the proposed solution for finding the closest paths. As the length of the retrieved path increases, indicating a more extensive search space, the time required for query response escalates accordingly. Despite this escalation, the solution's efficiency remains evident, particularly in scenarios necessitating the exploration of longer paths. This suggests the solution's viability and effectiveness in practical applications where extensive pathfinding is required, demonstrating its potential utility in various real-world contexts.

\section{Conclusion}
\label{sec:conclusion}
In conclusion, the need for secure graph encryption schemes is paramount in various applications, including navigation systems, social networks, and infrastructure modeling. Existing methods, while efficient, suffer from vulnerabilities such as access pattern and query pattern leakage, which compromise security and privacy. To address these challenges, we propose a novel graph encryption scheme that leverages Oblivious RAM (ORAM) and Trusted Execution Environments (TEE) to conceal both access pattern and query pattern leakage. Our scheme builds upon the foundation of existing state-of-the-art solutions like the GKT scheme, enhancing security while maintaining efficiency. Through evaluations on real-world datasets, we demonstrate the practicality and effectiveness of our solution in protecting sensitive graph data while allowing  secure shortest path queries.

\begin{credits}
\subsubsection{\ackname}
This work were supported by the french national research agency funded project AUTOPSY (grant no. ANR-20-CYAL-0008).  Additionally, part of
this work was done as part of IRT SystemX project PFS (Security of Smart Ports).

\end{credits}
%
%
%
\bibliographystyle{splncs04}
\bibliography{biblio}
\end{document}